\documentclass[twoside,11pt]{article}
\usepackage{jair, theapa, rawfonts}

\usepackage{environ}
\usepackage{tikz}
\usepackage{pgfplots}
\pgfplotsset{width=7.5cm,compat=1.12}
\usepgfplotslibrary{fillbetween}

\newsavebox{\measure@tikzpicture}
\NewEnviron{scaletikzpicturetowidth}[1]{%
  \def\tikz@width{#1}%
  \begin{lrbox}{\measure@tikzpicture}%
  \BODY
  \end{lrbox}%
  \pgfmathparse{#1/\wd\measure@tikzpicture}%
  \BODY
}

\usepackage[textwidth=0.9cm]{todonotes}

\newcommand{\commentout}[1]{}

\usepackage{amsmath,amssymb,amsthm}
\usepackage{graphicx}
\usepackage{subfigure,epstopdf}
\usepackage{xfrac}
\usepackage[T1]{fontenc}
\usepackage{lmodern}
\usepackage{microtype}

\newtheorem*{theorem}{Theorem}
\newtheorem{remark}{Remark}
\newcommand*{\sigx}{\ensuremath{\sigma^{x}}}

\newcommand*{\sigz}{\ensuremath{\sigma^{z}}}
\newcommand*{\x}{\mathbf{x}}
\newcommand*{\E}{\mathbb{E}}

\begin{document}

\title{Estimating the Density of States of Boolean Satisfiability Problems on Classical and Quantum Computing Platforms}

\author{\name Tuhin Sahai \; Anurag Mishra  \\
\addr United Technologies Research Center,\\
Berkeley, CA 94705.
\AND
\name Jos\'{e} Miguel Pasini \\
\addr United Technologies Research Center,\\
East Hartford, CT 06118.
\AND
\name Susmit Jha \\
\addr Computer Science Laboratory\\
SRI International,\\
Menlo Park, CA 94025.
}

% For research notes, remove the comment character in the line below.
% \researchnote

\maketitle

\begin{abstract}
Given a Boolean formula $\phi(x)$ in conjunctive normal form (CNF), the density of states counts the number of variable assignments that violate exactly $e$ clauses, for all values of $e$. Thus, the density of states is a histogram of the number of unsatisfied clauses over all possible assignments. This computation generalizes both maximum-satisfiability (MAX-SAT) and model counting problems and not only provides insight into the entire solution space, but also yields a measure for the \emph{hardness} of the problem instance. Consequently, in real-world scenarios, this problem is typically infeasible even when using state-of-the-art algorithms. While finding an exact answer to this problem is a computationally intensive task, we propose a novel approach for estimating density of states based on the concentration of measure inequalities. The methodology results in a quadratic unconstrained binary optimization (QUBO), which is particularly amenable to quantum annealing-based solutions. We present the overall approach and compare results from the D-Wave quantum annealer against the best-known classical algorithms such as the Hamze-de Freitas-Selby (HFS) algorithm and satisfiability modulo theory (SMT) solvers. 

%We also present results from the use of state-of-the-art satisfiability solvers for solving the QUBO problem.
% The density of states for a given problem not only provides insight into the complete solution space but also provides an accurate measure of the difficulty or hardness of the problem instance
\end{abstract}

\section{Introduction}
\label{introduction}
The density of states (DOS) for a given Boolean formula not only provides insight into the complete solution space but also serves as an accurate measure of the difficulty or hardness of the problem instance. 
The ability to compute DOS for optimization and feasibility problems has critical applications in system-requirements engineering of complex aerospace products. It provides a metric for requirements engineers to compare constraints, prescribed requirements~\cite{ferrante2016methodology,johnson1998178b,costello1995metrics}, and requirements decompositions~\cite{kirkman1998requirement}. 
%DoD specific
This computation is particularly germane to the design and optimization of complex aerospace systems~\cite{sommerville2005integrated}. 
Current classical methods for computing density of states~\cite{wang2001efficient,ermon2010computing,ermon2011flat} have limited scalability.
%are intractable for real world settings and, thus, rarely used. 
While the focus of this paper is on Boolean formulae, we note that constrained programming and feasibility problems can be easily mapped to equivalent Boolean satisfiability instances~\cite{walsh2000sat,tamura2009compiling}.

The DOS problem in the standard $k$-satisfiability ($k$-SAT) setting can be elucidated as follows: instead of deciding whether a given logical formula is satisfiable or not, one aims to compute the \emph{entire histogram of the number of clauses satisfied over all possible variable assignments}. Note  the DOS, for a given instance of an optimization or decision problem, captures its hardness (distributions with a low footprint for all satisfied clauses are harder to compute or satisfy). The DOS histogram sheds light on the fundamental nature of the feasible solution set and difficulty of solving the overall optimization problem. As stated earlier, these problems frequently arise when constructing complex systems for aerospace and defense applications~\cite{leveson2009software}.

The lack of methodical approaches that enable the comparison of competing safety-critical system requirements, while optimizing performance, stymie the development of next-generation complex systems. Note there are often multiple paths to decompose the overall system safety requirements down to subsystems requirements. Some of these decompositions may lead to costly design and redesign cycles to achieve desired levels of performance. Decompositions that have a higher DOS in the satisfiable range result in greater freedom to optimize performance and, consequently, result in quicker design cycles and fewer redesigns. 
% DoD specific
%Note that efficient requirement engineering can potentially reduce cost overruns and design delays in the development of complex aerospace systems such as the F-35~\cite{hughes2015lockheed}. 
The ability to quickly estimate the DOS of satisfiability problems will enable the specification engineer to ensure the prescribed requirements are satisfiable, internally consistent, and amenable to design space exploration very early in the design requirement step. 
% we would rather say, we chosen quantum because it fit into this problem
%Moreover, we believe that the DOS computation is an attractive test problem for comparing quantum annealers with classical computing platforms.

In this work, we aim to construct novel approaches for rapidly computing the DOS for a SAT problem (Boolean formula)~\cite{biere2009handbook}. Our approximate approach to estimate DOS of SAT instances exploits the concentration of measure inequalities~\cite{boucheron2013concentration}. These inequalities provide bounds on the tails of the distributions of random functions and have been used to construct the theory of generalization in machine learning~\cite{abu2012learning}, compute optimal bounds on uncertainty~\cite{owhadi2013optimal}, certify systems~\cite{leyendecker2010optimal}, compute bias of statistical estimators~\cite{gourgoulias2017biased}, 
and derive results in random matrix theory~\cite{tao2012topics}. 

In this paper, we make the following contributions:
\begin{enumerate}
\item We introduce a novel approach to estimate the density of states for SAT problems by using concentration of measures (McDiarmid's inequality), and  bound the deviation of the number of unsatisfied clauses (energy) from the expected (mean) number of unsatisfied clauses for uniformly distributed assignments. 
\item The deviation of the energy function from its expected value depends on its diameter (function variability), which can be computed by solving an optimization (maximization) problem~\cite{owhadi2013optimal}. We show this maximum deviation computation can be posed in the form of a quadratic unconstrained binary optimization (QUBO) that is particularly amenable to quantum annealers and results in tight bounds on the DOS histogram.
\item We demonstrate our approach on classical platforms by computing the diameter and associated concentration of measure bounds using  Selby's implementation~\cite{selby2013qubo,selby2014efficient} of the Hamze-de Freitas- Selby (HFS) algorithm~\cite{hamze2004fields}. 
\item We use satisfiability modulo theory (SMT) solvers~\cite{bjorner2015nuz,sebastiani2015optimathsat,dutertre2014yices} to solve the QUBO formulation as an alternative approach on classical platforms. The solutions from SMT solvers provide  tighter estimates but require significantly higher computational effort and do not scale.
\item We then compare the classical results to the computations on the D-Wave quantum annealer, a commercially available noisy intermediate-scale quantum (NISQ) device~\cite{preskill2018quantum}. We find the D-Wave machine provides higher-quality solutions when compared to the HFS algorithm, and scales better than SMT solvers. We further note the search for useful problems that are appropriate for present day NISQ devices is a very active area of research within quantum computing~\cite{preskill2018quantum}. We propose the DOS computation task as a potential test problem that can be used to \emph{benchmark current- and next-generation quantum annealers against their classical counterparts}.
\end{enumerate}
%to name a few. 

%assuming that $0$ or $1$ values for each variable are equally likely. 
%The expected number of unsatisfied clauses can be computed a priori for a given instance of the SAT. 

The rest of the paper is organized as follows. In section~\ref{sec:ksat}, we introduce the SAT problem and its associated density of states. We then discuss the existing state-of-the-art methods for computing DOS and highlight their limitations. Section~\ref{sec:conc} introduces concentration of measure inequalities in the context of DOS computations for the SAT problem. Section~\ref{sec:qubo} presents our novel formulation of the DOS problem as a QUBO for estimating the associated energy histogram. In section~\ref{sec:results}, we present results comparing state-of-the-art algorithms for computing density of states with our proposed concentration of measures approach on classical and quantum platforms. For the concentration of measures approach, we further present comparisons between the performance of the D-Wave machine with the HFS algorithm and the Z3 SMT solver.  We conclude in section~\ref{conclusion} by summarizing our results and presenting key challenges.
%We then conclude the paper with open questions and future work.

%------------------------------------------------

%------------------------------------------------
%\section{Background}
\section{Background}\label{sec:ksat}
\label{background}
Consider a $k$-SAT formula $\phi(x)$ of $N$ binary variables and $m$ clauses, $\phi:\{0,1\}^{N}\rightarrow\{0,1\}$, written
in the conjunctive normal form (CNF)~\cite{biere2009handbook} as follows,
\begin{equation}
\label{eq:sat}
\phi(x) = \bigwedge_{i=1}^m C_i = \bigwedge_{i=1}^m (x_{i_1} \lor x_{i_2} \lor\hdots\lor x_{i_{k}}),
\end{equation}
where $x_{i_l}$ is the $l^{\rm th}$ literal in clause $C_{i}$. A SAT formula is said to be 
satisfiable if there exists an assignment for the binary variables $\x$ such that $\phi(\x)=1\,\,          (\text{true})$. It is well known that the satisfiability problem is NP-complete~\cite{cook1971complexity}. A critical parameter associated with the satisfiability problem is the clause density $\alpha = m/N$~\cite{biere2009handbook}. In particular, the probability that a random $k$-SAT instance is satisfiable undergoes a phase transition as a function of $\alpha$ ($N\rightarrow\infty$)~\cite{xu2000exact,biere2009handbook}.  The MAX-SAT problem (and the corresponding weighted version)~\cite{krentel1988complexity,chieu2009relaxed} requires one to find that assignment (or assignments) that maximize the number (or the cumulative weights) of satisfied clauses. Consider a SAT formula $\phi$, then every assignment $x$ can be mapped to an ``energy'' $\Phi(x)$ such that,
\begin{equation}
  \label{eq:1}
  \Phi(\x) = \sum_{i=1}^{m} C_{i},
\end{equation}
where $C_i=1$, if the $i$-th clause evaluates to $\text{true}$. In other words, the goal under the MAX-SAT problem is to find the assignment for $\x$
%\mbox{vec}(x_{i_l})$ 
such that the number of satisfied clauses (or energy) is maximized. Using De Morgan's laws,
%along with the identities $x_1 \lor x_2 = x_{1}x_{2}$ and $\lnot x=1-x$,  
one can easily show that,
%\begin{align}
%  \label{eq:2}
%  x_{1} \lor x_{2} \lor x_{3} &= \lnot(\lnot x_{1} \land \lnot x_{2} \land \lnot x_{3}) \\
%  &= 1 - (1-x_{1})(1-x_{2})(1-x_{3})\ .
%\end{align}
%and hence
\begin{align}
  \label{eq:3}
  \Phi(\x) &= m-\displaystyle\sum_{i=1}^{m}\prod_{l=1}^{k} f(x_{i_l}),\\
  \text{where} f(x_{i_{l}})&=\begin{cases}
    x_{i_l},& \text{if } x_{i_l} \text{is negated in the clause}  \\
    (1-x_{i_1}),        & \text{otherwise}.
\end{cases}
\end{align}
Using the above formula, it is easy to see that if each literal $x_{i_{l}}$ were random with equal probability for values $\{0,1\}$, then the expected number of satisfied clauses is,
\begin{align}
\E[\Phi(\x)] &= \frac{m(2^{k} - 1)}{2^{k}}.
\end{align}
Thus, even though the satisfiability is NP-complete, a random assignment is expected to satisfy a large fraction of the clauses. Note that for the 3-SAT, the above formula reduces to,
\begin{align}
\E[\Phi(\x)] = 7m/8.
\end{align}
One can use this expected (mean) value of the number of satisfied clauses to estimate the DOS using concentration of measure inequalities (see section~\ref{sec:conc}). For SAT instances that arise from specific application domains (thus, not random), one can estimate the expected number of satisfied clauses by sampling over the independent variables in the Boolean formula.

%\subsection{Computation of density of States}
The DOS $d(e)$ of a SAT formula $\phi$ is equal to the number of assignments $\x$ for which
$\Phi(\x)=e$. In other words, it is the histogram of the number assignments as a function of $e$ satisfied clauses. Note the value of the number of satisfied clauses $e$ lies between $0$ and $m$ where
$m$ is the total number of clauses in the SAT formula. Following the terminology from the
physics community, we will also call this the energy of SAT formula. Since the total number of possible assignments is $2^N$, one can define the normalized density of states as follows,
\begin{equation}
  \label{eq:5}
  p(e) = \frac{d(e)}{2^{N}}\ .
\end{equation}
The normalized DOS acts as a discrete probability distribution.  Note 
it is not necessary that all energies $e$ have a valid assignment. For example, if the SAT formula cannot
be satisfied then $p(m)=0$. As explained previously, for a random 3-SAT, the mean of $p(e)$ vs $e$ is at $7m/8$~\cite{ermon2010computing}.

The density computation problem generalizes
computationally hard problems of
MAX-SAT and model counting~\cite{birnbaum1999good}.
The state-of-the-art algorithm for computing DOS is inspired by the Wang and Landau random walk algorithm~\cite{wang2001efficient}. In~\cite{ermon2010computing,ermon2011flat}, the authors propose an adaptive Markov Chain Monte Carlo (MCMC) approach called MCMC-FlatSAT that aims to sample from a steady-state distribution such that probability of a particular assignment $\sigma$ is inversely proportional to its DOS. Thus, the sampling approach effectively converges to a  flat-visit histogram (captured by a flatness parameter). In~\cite{ermon2010computing,ermon2011flat}, the authors test the algorithm on multiple benchmark examples. We use this method to compute the DOS for a series of random SAT instances and compare the resulting histograms to estimated DOS using our concentration of measure-based approach (as outlined in section~\ref{sec:conc}). 

Fig.~\ref{fig:DOS} describes a simple example of the DOS problem for a Boolean satisfiability problem with  $N=100$ and $\alpha=4.0$, and shows an example output of the MCMC-FlatSAT algorithm. In our experiments, we found that for $k$-SAT instances close to the phase transition~\cite{monasson1999determining}, the mixing times of the Markov chain~\cite{levin2017markov} increase significantly.
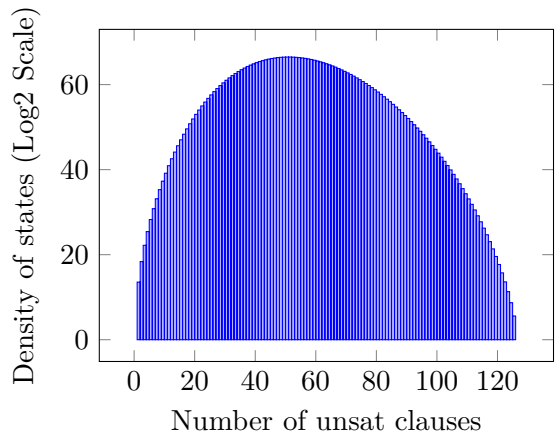
\begin{figure}[h]
%\centerline{\includegraphics[scale=0.2]{Figures/N100alpha400Inst10.eps}}
%\begin{scaletikzpicturetowidth}{\textwidth/2}
\centering
\begin{tikzpicture}%[scale=\tikzscale]
        \begin{axis}[
            width=0.5\textwidth,
            height=6cm,
            area style,
            xtick = {0,20,40,60,80,100,120,140},
            ylabel = {Density of states (Log2 Scale)},
            xlabel = {Number of unsat clauses},
            ]
        \addplot+[ybar interval,mark=no] plot coordinates { 
  (	1	,	13.55	)
(	2	,	18.37	)
(	3	,	22.21	)
(	4	,	25.43	)
(	5	,	28.27	)
(	6	,	30.82	)
(	7	,	33.15	)
(	8	,	35.29	)
(	9	,	37.29	)
(	10	,	39.16	)
(	11	,	40.92	)
(	12	,	42.56	)
(	13	,	44.13	)
(	14	,	45.61	)
(	15	,	47.01	)
(	16	,	48.33	)
(	17	,	49.59	)
(	18	,	50.79	)
(	19	,	51.92	)
(	20	,	53.00	)
(	21	,	54.02	)
(	22	,	54.98	)
(	23	,	55.89	)
(	24	,	56.76	)
(	25	,	57.58	)
(	26	,	58.35	)
(	27	,	59.08	)
(	28	,	59.76	)
(	29	,	60.41	)
(	30	,	61.02	)
(	31	,	61.58	)
(	32	,	62.11	)
(	33	,	62.61	)
(	34	,	63.07	)
(	35	,	63.50	)
(	36	,	63.89	)
(	37	,	64.26	)
(	38	,	64.59	)
(	39	,	64.89	)
(	40	,	65.17	)
(	41	,	65.42	)
(	42	,	65.64	)
(	43	,	65.83	)
(	44	,	66.00	)
(	45	,	66.14	)
(	46	,	66.25	)
(	47	,	66.35	)
(	48	,	66.42	)
(	49	,	66.46	)
(	50	,	66.48	)
(	51	,	66.48	)
(	52	,	66.46	)
(	53	,	66.42	)
(	54	,	66.36	)
(	55	,	66.27	)
(	56	,	66.17	)
(	57	,	66.05	)
(	58	,	65.91	)
(	59	,	65.74	)
(	60	,	65.56	)
(	61	,	65.36	)
(	62	,	65.15	)
(	63	,	64.91	)
(	64	,	64.66	)
(	65	,	64.39	)
(	66	,	64.10	)
(	67	,	63.79	)
(	68	,	63.46	)
(	69	,	63.12	)
(	70	,	62.76	)
(	71	,	62.39	)
(	72	,	61.99	)
(	73	,	61.58	)
(	74	,	61.16	)
(	75	,	60.71	)
(	76	,	60.26	)
(	77	,	59.78	)
(	78	,	59.29	)
(	79	,	58.78	)
(	80	,	58.25	)
(	81	,	57.71	)
(	82	,	57.15	)
(	83	,	56.57	)
(	84	,	55.97	)
(	85	,	55.36	)
(	86	,	54.73	)
(	87	,	54.08	)
(	88	,	53.42	)
(	89	,	52.73	)
(	90	,	52.03	)
(	91	,	51.31	)
(	92	,	50.56	)
(	93	,	49.80	)
(	94	,	49.02	)
(	95	,	48.21	)
(	96	,	47.39	)
(	97	,	46.54	)
(	98	,	45.68	)
(	99	,	44.78	)
(	100	,	43.87	)
(	101	,	42.93	)
(	102	,	41.96	)
(	103	,	40.97	)
(	104	,	39.94	)
(	105	,	38.89	)
(	106	,	37.81	)
(	107	,	36.69	)
(	108	,	35.53	)
(	109	,	34.34	)
(	110	,	33.12	)
(	111	,	31.84	)
(	112	,	30.52	)
(	113	,	29.15	)
(	114	,	27.73	)
(	115	,	26.24	)
(	116	,	24.69	)
(	117	,	23.07	)
(	118	,	21.36	)
(	119	,	19.58	)
(	120	,	17.69	)
(	121	,	15.72	)
(	122	,	13.60	)
(	123	,	11.30	)
(	124	,	8.70	)
(	125	,	5.58	)
(	126	,	1.39	)
         };
        \end{axis}
 \end{tikzpicture}
% \end{scaletikzpicturetowidth}
\caption{
This example Boolean satisfiability problem has $N=100$ variables, and the clause density ratio $\alpha$ is $4.0$. 
The x-axis is the number of UNSAT clauses and the y-axis provides a numerical value for the number of assignments with the corresponding number of UNSAT clauses. The DOS over all possible variable assignments are captured in the histogram. 
}\label{fig:DOS}
\end{figure}

\section{SAT and Concentration of measure inequalities}\label{sec:conc}
In this section, we
describe our approach based on a concentration of measure inequality for estimating DOS of satisfiability problems, and its relationship to the quadratic unconstrained binary optimization (QUBO) problem.

The concentration of measure phenomena bounds the deviation of functions of random variables around their mean~\cite{boucheron2013concentration}. There are a host of inequalities associated with various situations, see~\cite{boucheron2013concentration,tao2010254a} for more details. For our setting, we use McDiarmid's inequality~\cite{mcdiarmid1989method} summarized in the theorem below.
\begin{theorem}
(McDiarmid's inequality) Let $x_{1}, x_{2},\hdots, x_{N}$ be random variables taking values in the range $R_{1}, R_{2},\hdots, R_{N}$ and let $F: R_{1}\times R_{2}\times\hdots\times R_{N}\rightarrow  \mathbb{R}$ be a function with the property that if one freezes all but the $i$-th variable, then $F(x_{1}, x_{2},\hdots, x_{N})$ fluctuates by at most $D_{i}>0$,
\begin{align}
D_i = \sup_{x_{1},x_{2},\hdots, x_i,\hat{x}_{i},\hdots,x_{N}} &\left|F(x_{1}, x_{2}, \hdots, x_{i},\hdots,x_{N}))\right. \nonumber\\
& \left.- F(x_{1}, x_{2}, \hdots,\hat{x}_{i},\hdots, x_{N}))\right|.
\label{eq:dia}
\end{align}
Then the probability that $F$ deviates from its expected value is given by,
\begin{align}
\mathbb{P}\left[\left|F(x_{1}, x_{2},\hdots, x_{N}) - \mathbb{E}[F(x_{1}, x_{2},\hdots, x_{N})]\right| \ge \epsilon \right] \nonumber\\
\le C\exp( - c\frac{2\epsilon^2}{D^2}),
\label{eq:mcdiarmid}
\end{align}
where $D = \sqrt{\sum_i D_i^2}$ is called the diameter and $C, c$ are constants.

\end{theorem}
\begin{proof}
See~\cite{boucheron2013concentration,tao2010254a}.
\end{proof}

So instead of computing the DOS using the MCMC approach outlined in section~\ref{sec:ksat}, one can exploit McDiarmid's inequality to compute bounds on the histogram of number of satisfied (or unsatisfied) clauses. In the setting of the $k$-SAT problem, the $x_{1}, x_{2},\hdots, x_{N}$ in the McDiarmid's inequality are replaced by the variables $\x$ present in the logical formula.
%Note that in the $k$-SAT setting, some literal may occur in multiple clauses, thus, in the concentration of measure computations it is imperative that one ensures that $\x$ is a unique list of variables. %In other words,
That is,
$x_{1}, x_{2},\hdots, x_{N}$ are the unique set of Boolean variables that occur in the formula. The MCMC computation is now replaced by the set of optimization problems for computing the diameter as shown in Eqn.~\ref{eq:dia}.

For ease of presentation, we focus on the $3$-SAT problem instead of the generic $k$-SAT formulation. Polynomial time reductions from $k$-SAT to $3$-SAT make this translation nonrestrictive. Every $k$-SAT instance can be converted to a $3$-SAT instance by introducing additional (ancillary) variables. We now show that the diameter computations for the $3$-SAT problem give rise to a QUBO~\cite{boros2007local,rieffel2011quantum} problem that is particularly amenable to quantum annealers~\cite{kochenberger2014unconstrained}. %Note that we focus on the $3$-SAT problem as it is NP-hard and every $k$-SAT instance can be converted to a $3$-SAT instance by introducing additional variables~\cite{chancellor2016direct}.

\section{QUBO formulation for diameter computations} \label{sec:qubo}
To estimate $D_i$ in Eqn.~\ref{eq:dia} for the $3$-SAT setting, consider the following form for the SAT formula,
\[
\phi(\x) = \bigwedge_{i=1}^m C_i = \bigwedge_{i=1}^m (x_{i_1} \vee x_{i_2} \vee x_{i_3}),
\]
where $x_{i_l}$ is the $l$-th literal in clause $C_i$. It is easy to check that the number of satisfied clauses can be expressed as:
\begin{align}
\Phi(\x) = \sum_{i=1}^m C_i = \sum_{i=1}^m (&x_{i_1}  + x_{i_2} + x_{i_3} - x_{i_1}x_{i_2} \nonumber \\
& - x_{i_1}x_{i_3} - x_{i_2}x_{i_3} + x_{i_1}x_{i_2}x_{i_3}).
\end{align}
While this expression has a cubic term, the cubic term disappears when computing the diameter in Eqn.~\eqref{eq:dia} (shown later). We now state our central result that formulates the estimation of diameters $D_i$ needed for computing the DOS as a quadratic unconstrained Boolean optimization problem.

\begin{theorem}
The diameter $D_i$ for the variable $x_i$ in McDiarmid's inequality can be computed by solving the following optimization problem,
\begin{align*}
D_i = \max_{\x \setminus \x(i)} \left|\sum_{p\in S_i^+} \left[ 1 - x_{p_2} - x_{p_3} + x_{p_2} x_{p_3} \right]\right.\nonumber \\
\left. -\sum_{p\in S_i^-}  \left[ 1 - x_{p_2} - x_{p_3} + x_{p_2} x_{p_3} \right]\right|,
%\label{diameter_quadratic}
\end{align*}
where $S_i^+$ and $S_i^-$ are the sets of clauses in which $x(i)$ appears in direct and negated forms, respectively.
\end{theorem}

\begin{proof}
To compute the diameter $D_i$, pick the $i$-th variable of $\x$ denoted as $\x(i)$ and compute the worst-case variation in the number of satisfied clauses. Since $\Phi(\x)$ is a sum over different clauses, only those clauses that include $\x(i)$ in their literal set will contribute to the diameter. $S_i^+$ and $S_i^-$ are the sets of clauses in which $x(i)$ appears in direct and negated forms respectively, that is,
\begin{align*}
S_i^+ &= \{p: C_p = \x(i) \vee x_{p_2} \vee x_{p_3} \}, \\
S_i^- &= \{p: C_p = \neg \x(i) \vee x_{p_2} \vee x_{p_3} \},
\end{align*}
are the set of clauses in which the variable $\x(i)$ appears in the corresponding literal sets as either $\x(i)$ or $\neg \x(i)$, respectively. 
%Recall in a $3$-SAT instance a variable can only appear per clause.
Furthermore, because $\vee$ is commutative, we can assume without loss of generality that the variable $\x(i)$ appears as the first literal of the clause.  Therefore, the expression for the number of satisfied clauses in the $3$-SAT instance is: 
\begin{align}
\forall p \in S_i^+  \quad \Phi_p = &\x(i) + x_{p_2} + x_{p_3} \nonumber \\
 &- \x(i) x_{p_2} - \x(i) x_{p_3} - x_{p_2}x_{p_3} + \x(i) x_{p_2}x_{p_3}, \nonumber \\
\forall p \in S_i^-  \quad \Phi_p = &(1-\x(i)) + x_{p_2} + x_{p_3} \nonumber \\
&- (1-\x(i)) x_{p_2} - (1-\x(i)) x_{p_3} \nonumber \\
&- x_{p_2}x_{p_3} + (1-\x(i)) x_{p_2}x_{p_3}.
\label{eqn:raw}
\end{align}
Let $S_i^0 = \{1,\ldots,m\} \setminus (S_i^+ \cup S_i^-)$ be the set of clauses within which $\x(i)$ does not occur, thus,
\[
\Phi(x) = \sum_{p\in S_i^0} \Phi_p + \sum_{p\in S_i^+} \Phi_p + \sum_{p\in S_i^-} \Phi_p.
\]
The first term in the above sum is not affected by changing $\x(i)$ and does not contribute to the diameter and cancels in the subtraction in Eqn.~\ref{eq:dia}. Now, since $\x(i)$ can only take one of two values, $\{0,1\}$,
the number of clauses satisfied by setting $\x(i)$ to $1$ in $S_i^+$ is $ \displaystyle\sum_{p \in S_i^+} \left[1\right] $, and 
in $S_i^-$ is $ \displaystyle\sum_{p \in S_i^-} x_{p_2} + x_{p_3} - x_{p_2} x_{p_3} $ (computed using Eqn.~\ref{eqn:raw}). 
Symmetrically, the number of clauses satisfied by setting $\x(i)$ to $0$ in $S_i^+$ is $ \displaystyle\sum_{p \in S_i^-}\left[1\right] $, and  
in $S_i^+$ is $ \displaystyle\sum_{p \in S_i^+} x_{p_2} + x_{p_3} - x_{p_2} x_{p_3} $. $D_i$ is the maximum deviation between the two, that is,
\begin{align*}
D_i = \max_{\x \setminus \x(i)} \left|\sum_{p\in S_i^+} \left[ 1  \right] +  \sum_{p\in S_i^-}  \left[  x_{p_2} + x_{p_3} - x_{p_2} x_{p_3} \right] \right.\nonumber \\
\left. - \sum_{p\in S_i^+} \left[ x_{p_2} + x_{p_3} - x_{p_2} x_{p_3} \right] - \sum_{p\in S_i^-}  \left[ 1  \right]\right|.
\end{align*}
We get the following optimization problem to compute $D_i$ by collecting the terms for summing over $S_i^+$ and $S_i^-$,
\begin{align}
D_i = \max_{\x \setminus \x(i)} \left|\sum_{p\in S_i^+} \left[ 1 - x_{p_2} - x_{p_3} + x_{p_2} x_{p_3} \right]\right.\nonumber \\
\left. -\sum_{p\in S_i^-}  \left[ 1 - x_{p_2} - x_{p_3} + x_{p_2} x_{p_3} \right]\right|.
\label{diameter_quadratic}
\end{align}
\end{proof}

This result makes sense intuitively, because the expression inside each bracket is logically equivalent to $\neg (x_{p_2} \wedge x_{p_3})$, and if either of the other literals is true, the disjunctive clause $C_p$ remains false regardless of~$\x(i)$, and therefore does not contribute to the diameter.

The expression inside the absolute value in~\eqref{diameter_quadratic} is a quadratic form in $\x\setminus \x(i)$. Note that Eqn.~\ref{diameter_quadratic} can easily be cast into a purely quadratic form $\x^TQ\x$ as the linear terms can be absorbed into the diagonal of the matrix because $x_p = x_p^2$ for binary variables.
\begin{remark}
\label{rem1}
The computation for $D_i$ involves the maximization of an absolute value. To address the absolute value, we simply perform two separate maximizations as follows $\sup_\x | f(\x) | = \max \{\sup_\x f(\x), -\sup_\x f(\x)\}$. Thus, we compute the two maximizations and choose the larger result to obtain the diameter. Note that for a $3$-SAT instance with $N$ unique variables, one needs to perform $2N$ optimizations.
\end{remark}
\begin{remark}
Besides providing a novel approach for estimating the DOS of $k$-SAT problems, the diameter computation can be used to benchmark optimization algorithms and computing platforms. In particular, by comparing the value of the computed diameter by different approaches, one can quantify their performance. Higher diameter values correspond to ``better'' solutions of the optimization problem.
\end{remark}
\begin{remark}
Using the density of states, one can extract the probability of a random assignment being
\emph{at least} $\epsilon$ (in terms of energy or number of clauses satisfied) away from the average value $\bar{E}=\sum_{e=0}^{m}p(e)e$. Thus,
\begin{equation}
  \label{eq:rem2}
  P\left[\, \left|E - \bar{E}\right| \geq \epsilon\, \right] = \sum_{\bar{E}+\epsilon}^{m}p(e)\ .
\end{equation}
This quantity can be computed numerically from $d(e)$.
\end{remark}
In our experiments, we solved this QUBO formulation using quantum and classical computing methods. We describe these in detail in the next section.

\section{Results}\label{sec:results}
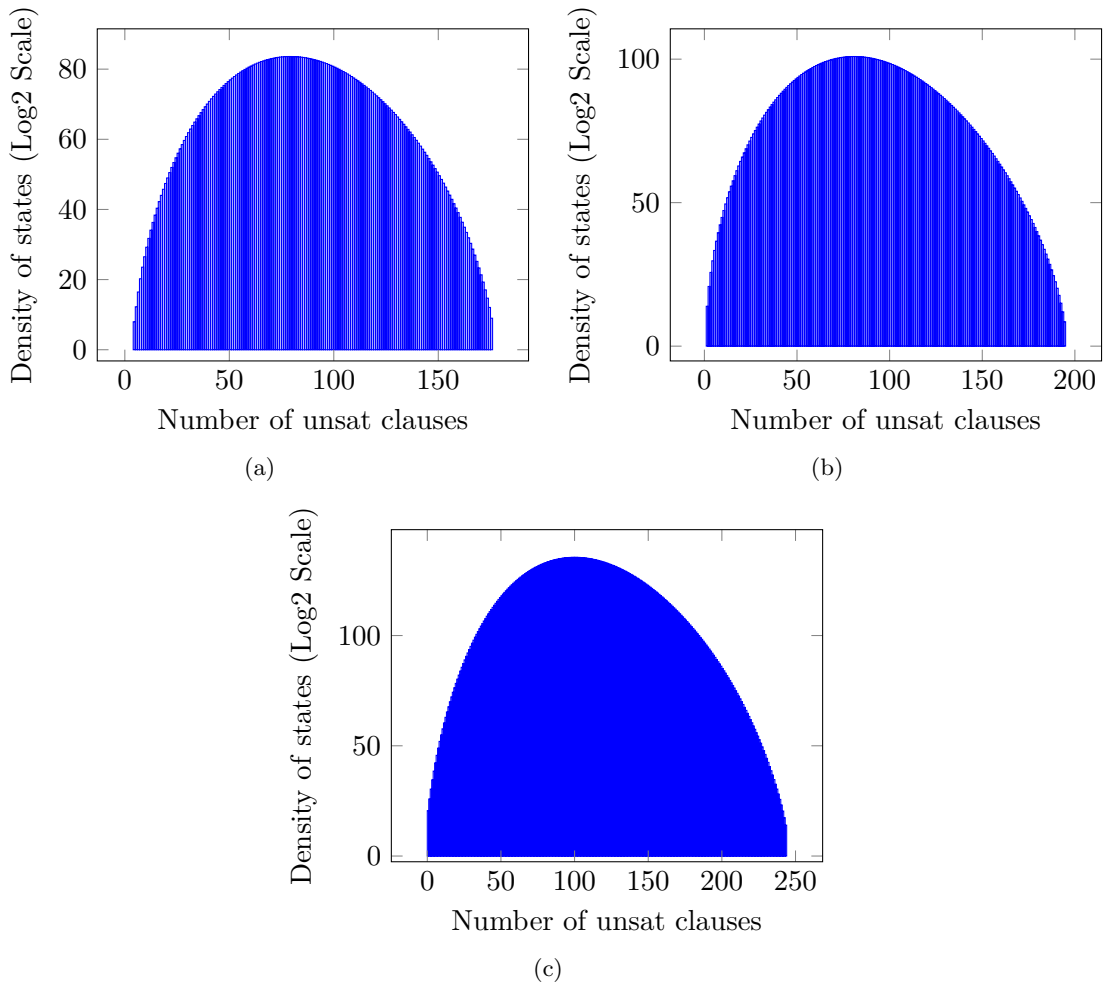
\begin{figure*}[!ht]
  \centering
  \subfigure[]{
  \begin{tikzpicture}%[scale=\tikzscale]
        \begin{axis}[
            width=0.48\textwidth,
            height=6cm,
            area style,
            xtick = {0,50,100,150,200},
            ylabel = {Density of states (Log2 Scale)},
            xlabel = {Number of unsat clauses},
            ]
        \addplot+[ybar interval,mark=no] plot coordinates { 
%(	1	,	0.00	)
%(	2	,	0.00	)
%(	3	,	0.00	)
(	4	,	7.98	)
(	5	,	12.25	)
(	6	,	16.46	)
(	7	,	20.23	)
(	8	,	23.54	)
(	9	,	26.51	)
(	10	,	29.23	)
(	11	,	31.74	)
(	12	,	34.09	)
(	13	,	36.30	)
(	14	,	38.38	)
(	15	,	40.35	)
(	16	,	42.23	)
(	17	,	44.02	)
(	18	,	45.74	)
(	19	,	47.40	)
(	20	,	48.98	)
(	21	,	50.50	)
(	22	,	51.97	)
(	23	,	53.38	)
(	24	,	54.73	)
(	25	,	56.03	)
(	26	,	57.29	)
(	27	,	58.50	)
(	28	,	59.67	)
(	29	,	60.79	)
(	30	,	61.88	)
(	31	,	62.92	)
(	32	,	63.93	)
(	33	,	64.90	)
(	34	,	65.84	)
(	35	,	66.74	)
(	36	,	67.61	)
(	37	,	68.46	)
(	38	,	69.27	)
(	39	,	70.05	)
(	40	,	70.81	)
(	41	,	71.53	)
(	42	,	72.23	)
(	43	,	72.90	)
(	44	,	73.55	)
(	45	,	74.18	)
(	46	,	74.77	)
(	47	,	75.35	)
(	48	,	75.90	)
(	49	,	76.43	)
(	50	,	76.94	)
(	51	,	77.43	)
(	52	,	77.89	)
(	53	,	78.34	)
(	54	,	78.76	)
(	55	,	79.16	)
(	56	,	79.55	)
(	57	,	79.91	)
(	58	,	80.26	)
(	59	,	80.59	)
(	60	,	80.89	)
(	61	,	81.18	)
(	62	,	81.46	)
(	63	,	81.71	)
(	64	,	81.95	)
(	65	,	82.17	)
(	66	,	82.38	)
(	67	,	82.57	)
(	68	,	82.74	)
(	69	,	82.89	)
(	70	,	83.03	)
(	71	,	83.16	)
(	72	,	83.26	)
(	73	,	83.36	)
(	74	,	83.44	)
(	75	,	83.50	)
(	76	,	83.55	)
(	77	,	83.58	)
(	78	,	83.60	)
(	79	,	83.60	)
(	80	,	83.59	)
(	81	,	83.57	)
(	82	,	83.53	)
(	83	,	83.48	)
(	84	,	83.41	)
(	85	,	83.33	)
(	86	,	83.24	)
(	87	,	83.13	)
(	88	,	83.01	)
(	89	,	82.88	)
(	90	,	82.73	)
(	91	,	82.57	)
(	92	,	82.40	)
(	93	,	82.21	)
(	94	,	82.01	)
(	95	,	81.80	)
(	96	,	81.58	)
(	97	,	81.34	)
(	98	,	81.09	)
(	99	,	80.83	)
(	100	,	80.55	)
(	101	,	80.26	)
(	102	,	79.96	)
(	103	,	79.64	)
(	104	,	79.32	)
(	105	,	78.98	)
(	106	,	78.63	)
(	107	,	78.26	)
(	108	,	77.88	)
(	109	,	77.49	)
(	110	,	77.09	)
(	111	,	76.68	)
(	112	,	76.26	)
(	113	,	75.81	)
(	114	,	75.36	)
(	115	,	74.90	)
(	116	,	74.42	)
(	117	,	73.93	)
(	118	,	73.43	)
(	119	,	72.91	)
(	120	,	72.38	)
(	121	,	71.84	)
(	122	,	71.28	)
(	123	,	70.72	)
(	124	,	70.13	)
(	125	,	69.54	)
(	126	,	68.93	)
(	127	,	68.31	)
(	128	,	67.67	)
(	129	,	67.03	)
(	130	,	66.36	)
(	131	,	65.68	)
(	132	,	64.99	)
(	133	,	64.28	)
(	134	,	63.56	)
(	135	,	62.82	)
(	136	,	62.06	)
(	137	,	61.29	)
(	138	,	60.51	)
(	139	,	59.71	)
(	140	,	58.88	)
(	141	,	58.05	)
(	142	,	57.19	)
(	143	,	56.31	)
(	144	,	55.42	)
(	145	,	54.51	)
(	146	,	53.58	)
(	147	,	52.62	)
(	148	,	51.64	)
(	149	,	50.64	)
(	150	,	49.62	)
(	151	,	48.56	)
(	152	,	47.49	)
(	153	,	46.38	)
(	154	,	45.24	)
(	155	,	44.08	)
(	156	,	42.88	)
(	157	,	41.65	)
(	158	,	40.39	)
(	159	,	39.09	)
(	160	,	37.74	)
(	161	,	36.35	)
(	162	,	34.93	)
(	163	,	33.45	)
(	164	,	31.92	)
(	165	,	30.34	)
(	166	,	28.70	)
(	167	,	26.98	)
(	168	,	25.19	)
(	169	,	23.32	)
(	170	,	21.36	)
(	171	,	19.29	)
(	172	,	17.07	)
(	173	,	14.70	)
(	174	,	12.06	)
(	175	,	8.99	)
(	176	,	4.70	)
         };
        \end{axis}
 \end{tikzpicture}
  }
  \subfigure[]{
  \begin{tikzpicture}%[scale=\tikzscale]
        \begin{axis}[
            width=0.48\textwidth,
            height=6cm,
            area style,
            xtick = {0,50,100,150,200},
            ylabel = {Density of states (Log2 Scale)},
            xlabel = {Number of unsat clauses},
            ]
        \addplot+[ybar interval,mark=no] plot coordinates { 
(	1	,	13.89	)
(	2	,	20.80	)
(	3	,	25.69	)
(	4	,	29.74	)
(	5	,	33.34	)
(	6	,	36.57	)
(	7	,	39.53	)
(	8	,	42.27	)
(	9	,	44.84	)
(	10	,	47.26	)
(	11	,	49.55	)
(	12	,	51.72	)
(	13	,	53.78	)
(	14	,	55.74	)
(	15	,	57.61	)
(	16	,	59.41	)
(	17	,	61.13	)
(	18	,	62.78	)
(	19	,	64.37	)
(	20	,	65.89	)
(	21	,	67.36	)
(	22	,	68.78	)
(	23	,	70.15	)
(	24	,	71.46	)
(	25	,	72.74	)
(	26	,	73.96	)
(	27	,	75.14	)
(	28	,	76.29	)
(	29	,	77.39	)
(	30	,	78.46	)
(	31	,	79.49	)
(	32	,	80.49	)
(	33	,	81.46	)
(	34	,	82.39	)
(	35	,	83.30	)
(	36	,	84.17	)
(	37	,	85.01	)
(	38	,	85.83	)
(	39	,	86.62	)
(	40	,	87.38	)
(	41	,	88.11	)
(	42	,	88.82	)
(	43	,	89.50	)
(	44	,	90.16	)
(	45	,	90.80	)
(	46	,	91.41	)
(	47	,	92.00	)
(	48	,	92.57	)
(	49	,	93.11	)
(	50	,	93.64	)
(	51	,	94.14	)
(	52	,	94.62	)
(	53	,	95.09	)
(	54	,	95.53	)
(	55	,	95.95	)
(	56	,	96.36	)
(	57	,	96.74	)
(	58	,	97.10	)
(	59	,	97.45	)
(	60	,	97.78	)
(	61	,	98.09	)
(	62	,	98.38	)
(	63	,	98.66	)
(	64	,	98.92	)
(	65	,	99.16	)
(	66	,	99.39	)
(	67	,	99.60	)
(	68	,	99.79	)
(	69	,	99.97	)
(	70	,	100.13	)
(	71	,	100.28	)
(	72	,	100.41	)
(	73	,	100.53	)
(	74	,	100.62	)
(	75	,	100.71	)
(	76	,	100.78	)
(	77	,	100.84	)
(	78	,	100.88	)
(	79	,	100.91	)
(	80	,	100.92	)
(	81	,	100.92	)
(	82	,	100.90	)
(	83	,	100.87	)
(	84	,	100.83	)
(	85	,	100.78	)
(	86	,	100.70	)
(	87	,	100.62	)
(	88	,	100.53	)
(	89	,	100.42	)
(	90	,	100.30	)
(	91	,	100.16	)
(	92	,	100.01	)
(	93	,	99.85	)
(	94	,	99.68	)
(	95	,	99.49	)
(	96	,	99.30	)
(	97	,	99.09	)
(	98	,	98.87	)
(	99	,	98.63	)
(	100	,	98.39	)
(	101	,	98.13	)
(	102	,	97.86	)
(	103	,	97.58	)
(	104	,	97.29	)
(	105	,	96.98	)
(	106	,	96.66	)
(	107	,	96.33	)
(	108	,	95.99	)
(	109	,	95.64	)
(	110	,	95.28	)
(	111	,	94.90	)
(	112	,	94.52	)
(	113	,	94.12	)
(	114	,	93.72	)
(	115	,	93.30	)
(	116	,	92.87	)
(	117	,	92.43	)
(	118	,	91.98	)
(	119	,	91.51	)
(	120	,	91.04	)
(	121	,	90.55	)
(	122	,	90.06	)
(	123	,	89.55	)
(	124	,	89.03	)
(	125	,	88.50	)
(	126	,	87.96	)
(	127	,	87.41	)
(	128	,	86.85	)
(	129	,	86.28	)
(	130	,	85.69	)
(	131	,	85.09	)
(	132	,	84.49	)
(	133	,	83.86	)
(	134	,	83.24	)
(	135	,	82.59	)
(	136	,	81.94	)
(	137	,	81.27	)
(	138	,	80.60	)
(	139	,	79.91	)
(	140	,	79.21	)
(	141	,	78.49	)
(	142	,	77.76	)
(	143	,	77.03	)
(	144	,	76.28	)
(	145	,	75.51	)
(	146	,	74.74	)
(	147	,	73.95	)
(	148	,	73.15	)
(	149	,	72.33	)
(	150	,	71.50	)
(	151	,	70.66	)
(	152	,	69.80	)
(	153	,	68.93	)
(	154	,	68.05	)
(	155	,	67.15	)
(	156	,	66.23	)
(	157	,	65.30	)
(	158	,	64.36	)
(	159	,	63.39	)
(	160	,	62.41	)
(	161	,	61.42	)
(	162	,	60.41	)
(	163	,	59.38	)
(	164	,	58.34	)
(	165	,	57.27	)
(	166	,	56.19	)
(	167	,	55.08	)
(	168	,	53.96	)
(	169	,	52.81	)
(	170	,	51.64	)
(	171	,	50.45	)
(	172	,	49.23	)
(	173	,	48.00	)
(	174	,	46.73	)
(	175	,	45.43	)
(	176	,	44.10	)
(	177	,	42.75	)
(	178	,	41.35	)
(	179	,	39.92	)
(	180	,	38.45	)
(	181	,	36.93	)
(	182	,	35.37	)
(	183	,	33.74	)
(	184	,	32.07	)
(	185	,	30.32	)
(	186	,	28.49	)
(	187	,	26.59	)
(	188	,	24.56	)
(	189	,	22.44	)
(	190	,	20.17	)
(	191	,	17.69	)
(	192	,	14.97	)
(	193	,	11.95	)
(	194	,	8.50	)
(	195	,	4.49	)
            };
            \end{axis}
 \end{tikzpicture}
  }
%  \subfigure[]{\includegraphics[scale=0.17]{Figures/N150alpha425Inst15.eps}}
  \subfigure[]{
    \begin{tikzpicture}%[scale=\tikzscale]
        \begin{axis}[
            width=0.48\textwidth,
            height=6cm,
            area style,
            xtick = {0,50,100,150,200,250},
            ylabel = {Density of states (Log2 Scale)},
            xlabel = {Number of unsat clauses},
            ]
        \addplot+[ybar interval,mark=no] plot coordinates { 
(	0	,	20.59	)
(	1	,	25.88	)
(	2	,	30.35	)
(	3	,	34.55	)
(	4	,	38.52	)
(	5	,	42.22	)
(	6	,	45.68	)
(	7	,	48.93	)
(	8	,	52.00	)
(	9	,	54.91	)
(	10	,	57.69	)
(	11	,	60.36	)
(	12	,	62.91	)
(	13	,	65.35	)
(	14	,	67.70	)
(	15	,	69.98	)
(	16	,	72.17	)
(	17	,	74.29	)
(	18	,	76.32	)
(	19	,	78.30	)
(	20	,	80.22	)
(	21	,	82.07	)
(	22	,	83.88	)
(	23	,	85.62	)
(	24	,	87.32	)
(	25	,	88.97	)
(	26	,	90.57	)
(	27	,	92.11	)
(	28	,	93.63	)
(	29	,	95.09	)
(	30	,	96.52	)
(	31	,	97.90	)
(	32	,	99.25	)
(	33	,	100.56	)
(	34	,	101.83	)
(	35	,	103.07	)
(	36	,	104.28	)
(	37	,	105.45	)
(	38	,	106.59	)
(	39	,	107.69	)
(	40	,	108.77	)
(	41	,	109.82	)
(	42	,	110.84	)
(	43	,	111.83	)
(	44	,	112.79	)
(	45	,	113.73	)
(	46	,	114.64	)
(	47	,	115.52	)
(	48	,	116.38	)
(	49	,	117.21	)
(	50	,	118.03	)
(	51	,	118.81	)
(	52	,	119.57	)
(	53	,	120.31	)
(	54	,	121.03	)
(	55	,	121.73	)
(	56	,	122.41	)
(	57	,	123.06	)
(	58	,	123.70	)
(	59	,	124.31	)
(	60	,	124.90	)
(	61	,	125.48	)
(	62	,	126.03	)
(	63	,	126.57	)
(	64	,	127.09	)
(	65	,	127.59	)
(	66	,	128.07	)
(	67	,	128.54	)
(	68	,	128.99	)
(	69	,	129.42	)
(	70	,	129.83	)
(	71	,	130.23	)
(	72	,	130.61	)
(	73	,	130.98	)
(	74	,	131.32	)
(	75	,	131.66	)
(	76	,	131.98	)
(	77	,	132.28	)
(	78	,	132.57	)
(	79	,	132.84	)
(	80	,	133.10	)
(	81	,	133.35	)
(	82	,	133.58	)
(	83	,	133.80	)
(	84	,	134.00	)
(	85	,	134.19	)
(	86	,	134.36	)
(	87	,	134.52	)
(	88	,	134.67	)
(	89	,	134.81	)
(	90	,	134.93	)
(	91	,	135.04	)
(	92	,	135.14	)
(	93	,	135.22	)
(	94	,	135.30	)
(	95	,	135.36	)
(	96	,	135.41	)
(	97	,	135.44	)
(	98	,	135.47	)
(	99	,	135.48	)
(	100	,	135.48	)
(	101	,	135.47	)
(	102	,	135.45	)
(	103	,	135.42	)
(	104	,	135.37	)
(	105	,	135.32	)
(	106	,	135.25	)
(	107	,	135.17	)
(	108	,	135.08	)
(	109	,	134.99	)
(	110	,	134.88	)
(	111	,	134.76	)
(	112	,	134.62	)
(	113	,	134.48	)
(	114	,	134.33	)
(	115	,	134.17	)
(	116	,	133.99	)
(	117	,	133.81	)
(	118	,	133.62	)
(	119	,	133.41	)
(	120	,	133.21	)
(	121	,	132.98	)
(	122	,	132.75	)
(	123	,	132.51	)
(	124	,	132.25	)
(	125	,	131.99	)
(	126	,	131.72	)
(	127	,	131.44	)
(	128	,	131.15	)
(	129	,	130.85	)
(	130	,	130.54	)
(	131	,	130.22	)
(	132	,	129.89	)
(	133	,	129.55	)
(	134	,	129.21	)
(	135	,	128.85	)
(	136	,	128.49	)
(	137	,	128.11	)
(	138	,	127.73	)
(	139	,	127.34	)
(	140	,	126.93	)
(	141	,	126.52	)
(	142	,	126.10	)
(	143	,	125.67	)
(	144	,	125.23	)
(	145	,	124.79	)
(	146	,	124.33	)
(	147	,	123.86	)
(	148	,	123.39	)
(	149	,	122.90	)
(	150	,	122.41	)
(	151	,	121.91	)
(	152	,	121.39	)
(	153	,	120.87	)
(	154	,	120.34	)
(	155	,	119.80	)
(	156	,	119.25	)
(	157	,	118.69	)
(	158	,	118.13	)
(	159	,	117.55	)
(	160	,	116.97	)
(	161	,	116.37	)
(	162	,	115.76	)
(	163	,	115.15	)
(	164	,	114.53	)
(	165	,	113.90	)
(	166	,	113.26	)
(	167	,	112.60	)
(	168	,	111.95	)
(	169	,	111.27	)
(	170	,	110.59	)
(	171	,	109.90	)
(	172	,	109.20	)
(	173	,	108.49	)
(	174	,	107.77	)
(	175	,	107.04	)
(	176	,	106.30	)
(	177	,	105.55	)
(	178	,	104.79	)
(	179	,	104.01	)
(	180	,	103.23	)
(	181	,	102.44	)
(	182	,	101.63	)
(	183	,	100.81	)
(	184	,	99.99	)
(	185	,	99.15	)
(	186	,	98.30	)
(	187	,	97.44	)
(	188	,	96.57	)
(	189	,	95.68	)
(	190	,	94.78	)
(	191	,	93.87	)
(	192	,	92.95	)
(	193	,	92.01	)
(	194	,	91.07	)
(	195	,	90.10	)
(	196	,	89.12	)
(	197	,	88.13	)
(	198	,	87.13	)
(	199	,	86.11	)
(	200	,	85.07	)
(	201	,	84.02	)
(	202	,	82.95	)
(	203	,	81.87	)
(	204	,	80.77	)
(	205	,	79.65	)
(	206	,	78.51	)
(	207	,	77.36	)
(	208	,	76.19	)
(	209	,	74.99	)
(	210	,	73.78	)
(	211	,	72.55	)
(	212	,	71.30	)
(	213	,	70.03	)
(	214	,	68.74	)
(	215	,	67.42	)
(	216	,	66.07	)
(	217	,	64.71	)
(	218	,	63.32	)
(	219	,	61.90	)
(	220	,	60.45	)
(	221	,	58.99	)
(	222	,	57.49	)
(	223	,	55.96	)
(	224	,	54.39	)
(	225	,	52.80	)
(	226	,	51.18	)
(	227	,	49.51	)
(	228	,	47.81	)
(	229	,	46.08	)
(	230	,	44.30	)
(	231	,	42.48	)
(	232	,	40.60	)
(	233	,	38.69	)
(	234	,	36.71	)
(	235	,	34.67	)
(	236	,	32.55	)
(	237	,	30.38	)
(	238	,	28.09	)
(	239	,	25.70	)
(	240	,	23.14	)
(	241	,	20.41	)
(	242	,	17.42	)
(	243	,	14.02	)
(	244	,	9.96	)        
            };
            \end{axis}
 \end{tikzpicture}
  }
  \caption{Density of states for random satisfiability instances with varying size and clause density. The X-axis is the number of unsat clauses, Y-axis is the DOS showing number of assignments in log scale. (a) $N=125,\,\,\alpha = 5.0$, (b) $N=150,\,\, \alpha=4.25$, (c) $N= 200,\,\,\alpha = 4$.}
  \label{fig:dos_mcmc}
\end{figure*}

\begin{figure*}[!htp]
  \centering
  \subfigure[]{\includegraphics[scale=0.17]{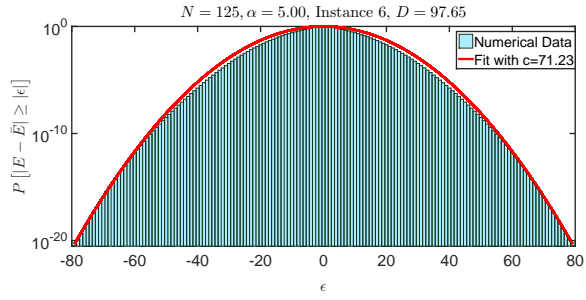}}
  \subfigure[]{\includegraphics[scale=0.17]{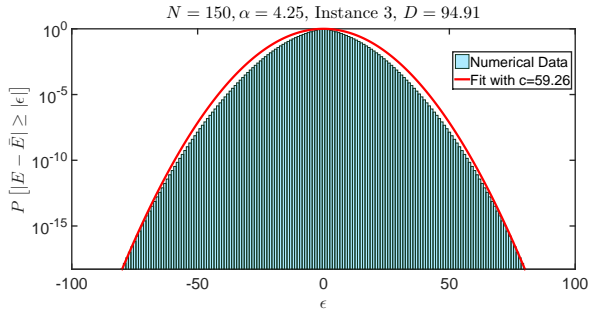}}
  \subfigure[]{\includegraphics[scale=0.17]{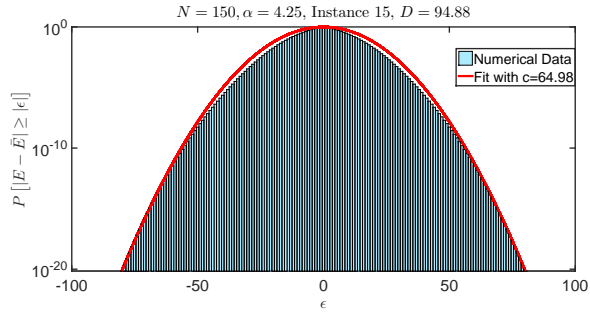}}
  \subfigure[]{\includegraphics[scale=0.17]{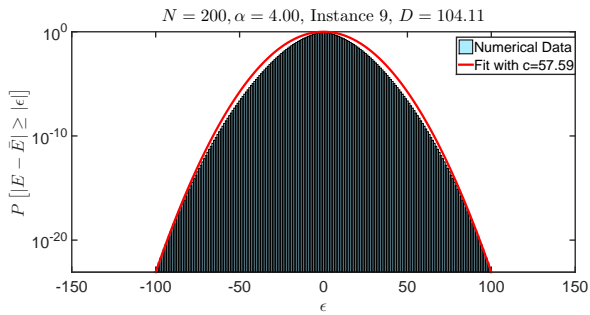}}
  \caption{Comparison of the density of states computed by MCMC and the concentration of measure bounds.\label{fig:dos_conc} }
\end{figure*}

\begin{figure}[!htp]
 \centering
\includegraphics[scale=0.18]{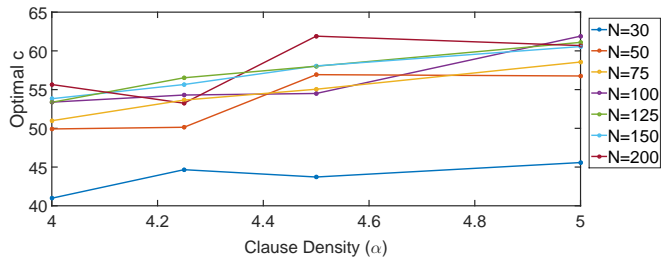}
  \caption{Concentration of measure constant as a function of $\alpha$. The different lines correspond to different values of N.\label{fig:c_alpha}}
\end{figure}

\begin{figure*}[!htp]
  \centering
  \subfigure[Instances for $N=100,\,\, \alpha=4.00$]{\includegraphics[scale=0.15]{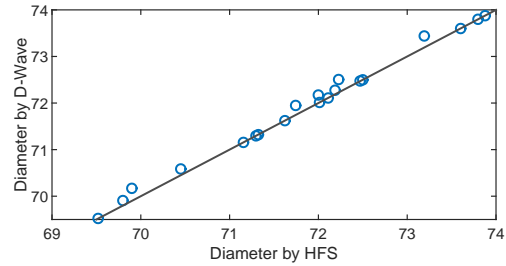}}
  \subfigure[Instances for $N=125,\,\, \alpha=5.00$]{\includegraphics[scale=0.15]{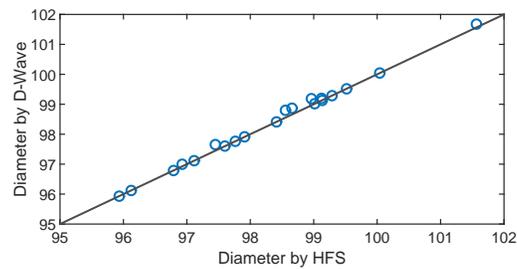}}
  \subfigure[Instances for $N=150,\,\,\alpha=4.25$]{\includegraphics[scale=0.15]{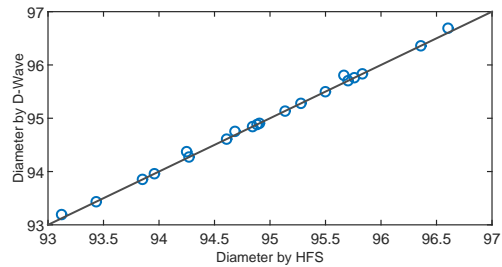}}
  \subfigure[Instances for $N=200,\,\,\alpha=4.00$]{\includegraphics[scale=0.15]{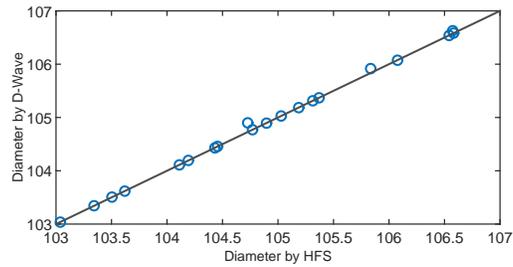}}
  \caption{Comparison of diameters computed by the D-Wave quantum annealer and HFS algorithm.\label{fig:comparisondwave}}
\end{figure*}

\begin{figure*}[!htp]
  \centering
  \subfigure[Instances for $N=150\,\,\alpha=4.25$ ]{\includegraphics[scale=0.15]{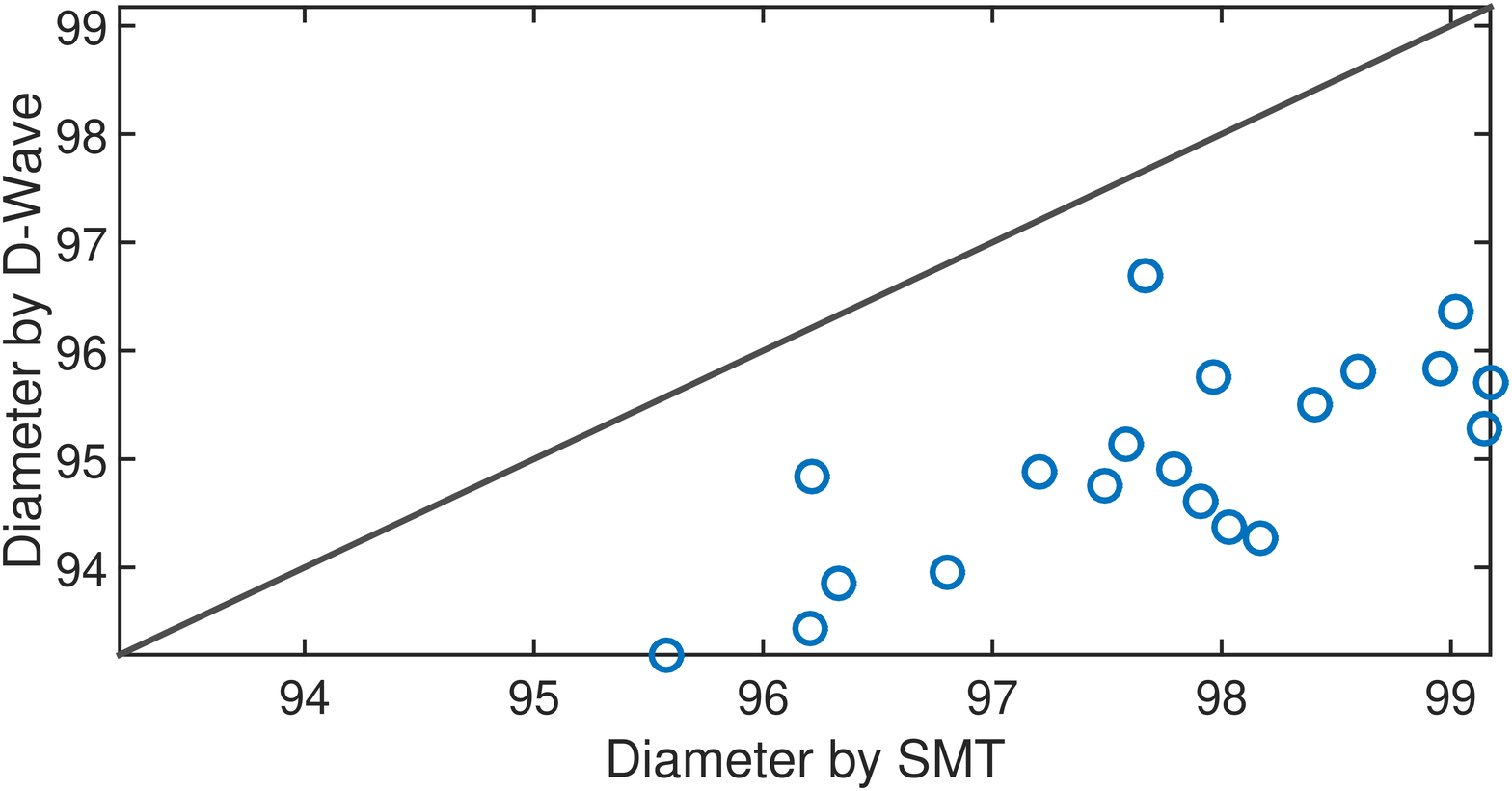}}
    \subfigure[Instances for $N=200\,\,\alpha=4.00$]{\includegraphics[scale=0.15]{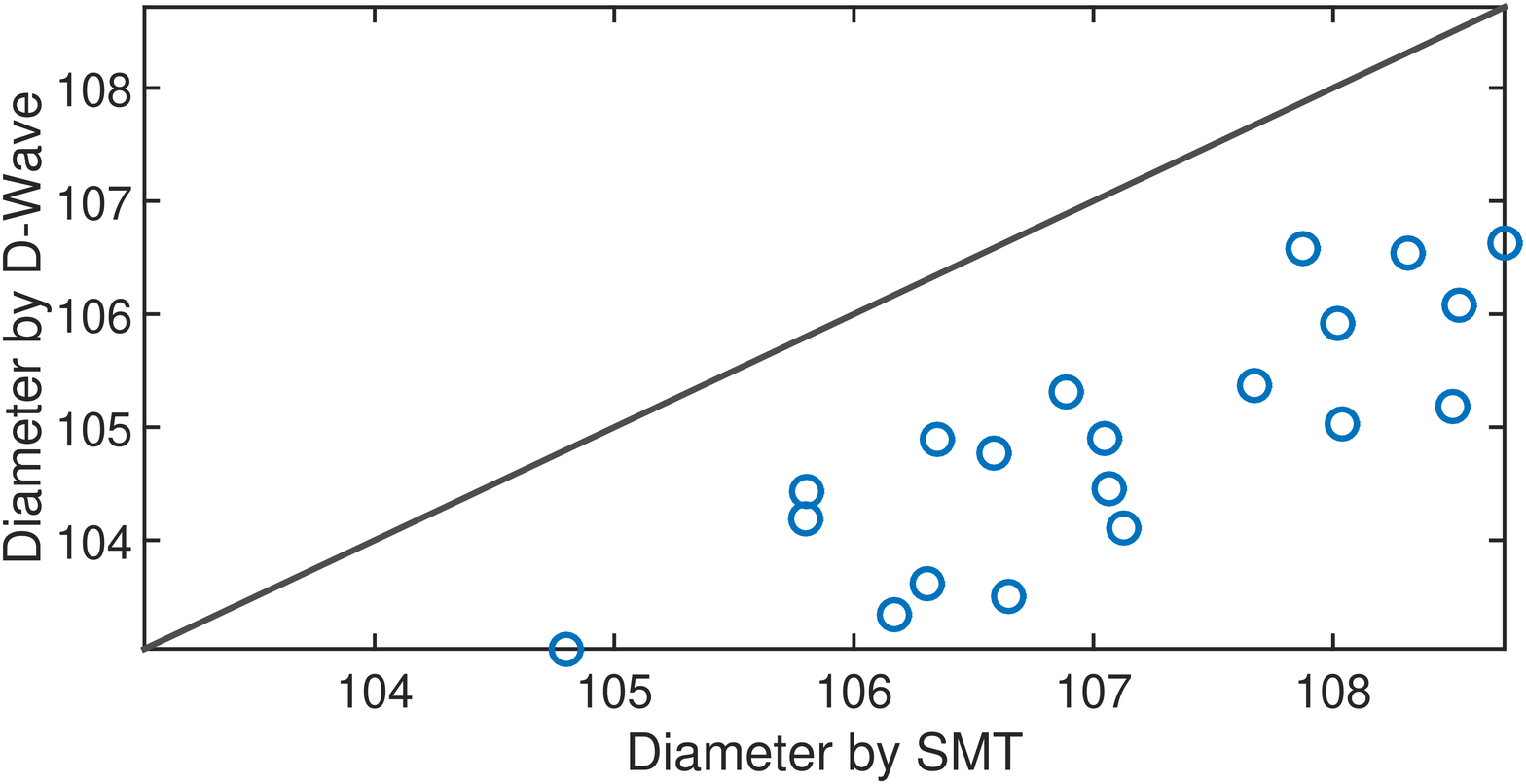}}
  \caption{Comparison of diameters computed by the D-Wave quantum annealer and SMT solver.\label{fig:comparisondwaveSMT}}
\end{figure*}

%\begin{figure*}[!htp]
  %\centering
  %\subfigure[]{\includegraphics[scale=0.6]{Figures/N150alpha4_25SMT.png}}
 %   \subfigure[]{\includegraphics[scale=0.6]{Figures/N200alpha4SMT.png}}
 % \subfigure[]{\includegraphics[scale=0.6]{Figures/N400alpha4SMT.png}}
%  \caption{Comparison of diameters computed by the D-Wave quantum annealer and SMT Solver .\label{fig:comparisondwaveSMT}}
%\end{figure*}

In our experiments, we attempt to answer the following questions.
\begin{itemize} 
\item How does the proposed DOS approach of using concentration of measures and QUBO compare with the baseline MCMC-FlatSAT approach~\cite{ermon2010computing,ermon2011flat}?
\item How do the quantum and classical implementations of the proposed DOS approach compare with one another?
\end{itemize}
%We now present results on estimating the density of states using the concentration of measure inequality based approach on classical and quantum computing platforms. We implemented the MCMC-FlatSAT approach~\cite{ermon2010computing,ermon2011flat} to compute the ground truth for the density of state histograms.

To analyze the performance of the proposed approach, we generate $20$ random $3$-SAT instances for every possible combination of the following sizes ($N=\left[30, 50, 75, 100, 125, 150, 200\right]$) and clause density ($\alpha=\left[4.0, 4.25, 4.5, 5\right]$). As mentioned earlier, although the proposed technique can also be easily applied to non-random SAT instances, the choice of random SAT instances allows variation from easy-to-hard problems. Our choice of $\alpha$ values span the phase transition at $\alpha\approx 4.24$ that demarcates ``easy'' and ``hard'' instances of the satisfiability problem~\cite{monasson1999determining}. Thus, in total, we generate $560$ random $3$-SAT instances, and for each instance we compute the baseline DOS using MCMC-FlatSAT. These results are then compared with the proposed concentration of measure inequality approach. Additionally, the $2N$ diameter computations for each instance are performed classically using the HFS algorithm and the D-Wave quantum annealer. The performance of the D-Wave device is then compared with the classical results. We also implemented an SMT-based optimization approach on classical platforms, and compared the D-Wave results with the standard classical solver.

\subsection{MCMC-FlatSAT results}
We implemented the MCMC-FlatSAT~\cite{ermon2010computing,ermon2011flat} algorithm in C++. 
%The code was run on a linux cluster node with $64$ GB RAM. 
Depending on the mixing time of the Markov chain~\cite{levin2017markov}, there was a large variation in the performance of the code. The computation time ranged from hours to several days (in some instances the code took $3-10$ days to converge). The computations were performed for all of the $560$ instances as outlined above. A few instances of the resulting density of states are shown in Fig.~\ref{fig:dos_mcmc}. In general, the MCMC approach was found to have a high computational cost.
%Given the high computational cost of generating the DOS histograms using the MCMC approach, we now use the concentration of measure approach (we use the D-Wave quantum annealer and compare the results to classical techniques).

\subsection{SMT results}
We used the Z3 SMT solver~\cite{bjorner2015nuz} to encode the QUBO problem as a bitvector problem exploiting the fixed range of discrete values that can be taken by the diameters. The resulting problem is a pseudo-Boolean optimization problem that we solve iteratively using satisfiability solving by binary search (between $0$ and the number of clauses in which the variable occurs) over the optimization goal. We allow the SMT solver a timeout of $100$ seconds for every trial to find a larger diameter. The $560$ instances took $8$ days to compute. The SMT solver found better solutions for the QUBO compared to quantum annealing, and hence placed more accurate bounds on the histogram. However, the scalability declined sharply with the increase in the number of variables. In particular, we found that the Z3 solver took multiple hours to days to complete several instances. The diameters computed using the SMT solver can be seen in Fig.~\ref{fig:comparisondwaveSMT} for a few values of $N$ and $\alpha$. 

\subsection{Quantum annealing results}
We use the D-Wave 2X (DW2X) annealer~\cite{johnson_2011_dwave} located at the USC Information
Science Institute in Marina del Rey as our quantum platform for computations. This DW2X processor is an $1152$-qubit quantum annealing
device made  using superconducting flux
qubits~\cite{bunyk_2014_dwave}. It has $1098$ functional qubits that 
function at $12$ mK. The annealer implements the transverse Ising
Hamiltonian,
\begin{equation}
\label{eq:dwave}
H(s) = A(s) \sum_i \sigx_i + B(s) \left(\sum_i h_i \sigz_i + \sum_{ij} J_{ij} \sigz_i \sigz_j \right),
\end{equation}
where $s=t/t_{\rm f}$ is the normalized time, $t_{f}$ is the total evolution time, and
$A(s)$ and $B(s)$ are the annealing schedules that modulate the transverse field and Ising
field strength, respectively. The total annealing time $t_{\mathrm{f}}$ can be set in the range
$[5,2000]\ \mu$s. The coupling strengths $J_{ij}$ between
qubits $i$ and $j$ can be set in the range $[-1,1]$, and the local fields $h_i$ can be set
in the range $[-2,2]$. Initially, $A(0)\gg B(0)$ and the system starts in the
superposition of all possible computational states. During the evolution from $s=0$ to
$s=1$, the transverse field is reduced and the Ising field strength is increased such that
$A(1) \ll B(1)$. If $t_f$ is large enough, the adiabatic theorem~\cite{born1928beweis}
guarantees that the final state of the system will be the ground state of $H(s=1)$. The device has been used for machine learning~\cite{biamonte2017quantum,adachi2015application,mott2017solving}, image recognition~\cite{neven2008image}, and combinatorial optimization~\cite{ushijima2017graph,mcgeoch2013experimental,neukart2017traffic,venturelli2015quantum} to name a few.

We use the above platform to compute the diameters for all the $560$ instances of random satisfiability problems and compared the results to MCMC-FlatSAT. As noted in remark~\ref{rem1}, each instance of an $N$-dimensional $3$-SAT problem gives rise to $2N$ optimizations for $D_i$. We chose the smallest possible annealing time $t_{f}=5\,\mu$s. For each QUBO instance of this study, we did
$1000$ readouts with $10$ gauge transforms
each~\cite{boixo_2014_evidencequantum}. Additional details of this particular process can
be found in, for example, Ref.~\cite{mishra_2018_finitetemperature}. Note that the total wall clock time to optimize each instance, which includes overheads such as initializing the qubits and measurements, was $\approx0.1$ second. Additionally, $1000$ readouts are on the low side;  however, we were restricted due to the sheer number of QUBOs ($\approx 120,000$ instances) coupled with limited affiliate time on the DW2X annealer. 

We map each diameter computation to a QUBO. 
\begin{align}
  \label{eq:diameter_quadratic}
  \begin{split}
    Q_{i}(\vec{x}) = &\sum_{p\in S_i^+} \left[ 1 - x_{p_2} - x_{p_3} + x_{p_2} x_{p_3} \right]  \\
    -&\sum_{p\in S_i^-} \left[ 1 - x_{p_2} - x_{p_3} + x_{p_2} x_{p_3} \right].  
  \end{split}
\end{align}

The size of this QUBO problem depends on the size of the sets $S_{i}^{+}$ and $S_{i}^{-}$
(see Section~\ref{sec:qubo} of the paper). If the SAT problem has $N$ variables and $\alpha$ clause density,
the number of clauses $M=N\alpha$ is a loose upper bound on the size of these QUBO
problems. Since the clauses can contain arbitrary variables, for which DW2X has a finite
connectivity graph, we need to find a minor embedding of the D-Wave graph that can fit this
QUBO problem~\cite{choi_2008_embed1,choi_2011_embed2}. In such embedding, each
$x_{p_i}$ in Eqn.~\ref{eq:diameter_quadratic} is represented by a chain of physical qubits 
connected via an ferromagnetic couplings. We used the \texttt{sapiFindEmbedding} function
provided by the D-Wave application program interfaces (API) to find such embeddings. We used the \texttt{sapiEmbedProblem}
function to submit the jobs to the processor and \texttt{sapiUnembedAnswer} function with
the \texttt{minimize\_energy} option to optimally decode the embeddings back to the
variables $\vec{x}$. We used the heuristic ferromagnetic chain coupling provided by the API. To find higher-quality solutions, one can optimize this ferromagnetic coupling
value such that the chain of physical qubits representing each variable is
consistent at the end of the anneal. Thus, our results provide a lower bound on the diameter.  Potentially, one may be able to obtain improved results by performing the actions suggested above and optimizing the annealing process.

%We perform the $D_i$ optimization on the D-Wave quantum annealer by give the machine $0.1$ secs of wall clock %time. This computation is repeated $1000$ times and the best solution is saved, while the rest are discarded.
After computing all the $D_i$'s for a given instance, we can plot the concentration of measure bounds for the DOS. Note that in McDiarmid's inequality (Eqn.~\ref{eq:mcdiarmid}), two constants appear that can be used to make the bounds on the DOS tight. In particular, we find $C=1$ and $c = 56.16 - 12.08\exp(-0.07(N-29.78)) + 6.88(\alpha -4.46) $ give rise to very close approximations of the density of states in the range of $30\leq N \leq 200$ (as shown in Fig.~\ref{fig:dos_conc}). These parameters were computed using the Broyden-Fletcher-Goldfarb-Shanno (BFGS) algorithm~\cite{Cit:BFGS}. The functional form for $c$ was obtained empirically by finding the best $c$ for each of the $560$ instances and performing regression with respect to $N$ and $\alpha$ (see Fig.~\ref{fig:c_alpha}).

\subsection{Comparison of D-Wave and HFS}
We repeat the $D_i$ computations for each of the $560$ instances by forcing the classical device to compute the best possible solution using the HFS algorithm. We again run the computation for $0.1$ secs, repeated $1000$ times for each $D_i$. The best $D_i$ is saved and the rest are discarded. We then compare the $D_i$ values obtained using quantum annealing with those computed classically. Note that higher diameter values correspond to ``better'' solutions as they correspond to higher quality solutions of the QUBO. Note we intend to conduct a comprehensive benchmarking study for the D-Wave quantum annealer~\cite{albash2018demonstration} (using our DOS framework) in future work.

Out of the $560$ random satisfiability instances, the D-Wave quantum computer computes higher quality solutions (higher diameter values) in $306$ instances. However, as shown in Fig.~\ref{fig:comparisondwave}, the D-Wave provided a marginal improvement on the diameter values. In particular, we found that the average solution computed by the D-Wave machine was around $0.7\%$ higher than the HFS algorithm. The most favorable result for the D-Wave was the computation of a solution that was $7\%$ better than the HFS algorithm. Whether this improvement holds for larger instances remains to be seen and will be tested in higher qubit settings. We would like to point out, however, that in no instance did the HFS algorithm find a higher quality solution when compared to the D-Wave machine. As shown in Fig.~\ref{fig:comparisondwaveSMT}, the SMT solver does find significantly better solutions than the D-Wave machine. Note that the computational cost of the solver is significantly higher (taking hours to days to compute the diameter of some instances).

\section{Conclusion}
\label{conclusion}
In this work, we have developed a novel concentration of measure inequalities-based approach for estimating the density of states of the $k$-SAT problem. Existing state-of-the-art Markov Chain Monte Carlo methods for computing density of states are stymied by computational intractability. Our approach provides estimates for the density of states histogram by converting the problem into a set of optimization problems that bound the maximum variability of the cost function (called ``diameters''). For the $3$-SAT, these optimization problems elegantly reduce to quadratic unconstrained binary optimizations, thereby making them amenable for commercially available quantum annealers such as the D-Wave machine.

To test our approach, we computed diameters for a range of random instances that span the phase transition of the $3$-SAT. We used the D-Wave quantum annealer and computed bounds on the density of states for those problem instances. We find that by tuning a single parameter (exponent), one can \emph{accurately estimate the entire density of states histogram}, orders of magnitude faster than state-of-the-art Markov Chain Monte Carlo methods. Additionally, we compared the D-Wave diameter computations to equivalent classical techniques such as SMT solvers and the HFS algorithm. We find that the D-Wave machine provides a marginal improvement in the diameter computations over the HFS algorithm. In other words, the D-Wave annealer computes solutions that correspond to larger values of the diameter. The SMT solver works particularly well for small problem instances. However, we find this approach does not scale well and is not competitive (from a computational standpoint) for larger instances of the 3-SAT. 

In summary, we propose a new approach for estimating density of states of the $k$-SAT problem that can be implemented on both classical and quantum platforms. Moreover, the problem is a particularly interesting test for comparing quantum platforms (annealers and other noisy intermediate-scale quantum devices) to classical computation. This is because the diameter values provide a real number metric for comparison. In other words, the quality of solution is a real number as opposed to the standard satisfiability tests for benchmarking quantum devices that yield inconclusive results of the form ``no satisfiable assignments found'' for most instances. We hope this problem and the outlined approach can be used to analyze complex aerospace system requirements from a satisfiability standpoint as well as to test emerging quantum platforms against their classical counterparts. DOS estimation can be used for probabilistic inference and, thus, an efficient quantum algorithm for the density of states estimation will enable development of quantum artificial intelligence. 

In future work, we plan to expand the use of concentration of measure inequalities for a wider set of problems and compare the above approaches with simulated annealing-based methods~\cite{aarts1988simulated}.
\section{Acknowledgements}
The authors thank Federico Spedalieri and Daniel Lidar of the University of Southern California-Information Sciences Institute (USC-ISI) for discussions and suggestions. The authors thank Lockheed Martin's quantum computing team (Christopher Elliott, Greg Tallant, and Kristen Pudenz) for discussions related to the approach and generously providing affiliate access to their D-Wave quantum annealer. Dr. Jha acknowledges support from the US Army Research Laboratory Cooperative Research Agreement W911NF-17-2-0196, and National Science Foundation (NSF) grants \#1750009, and \#1740079. The views, opinions and/or findings expressed are those of the author(s) and should not be interpreted as representing the official views or policies of the Department of Defense or the U.S. Government.

%----------------------------------------------------------------------------------------
%	REFERENCE LIST
%----------------------------------------------------------------------------------------

%\clearpage
\vskip 0.2in
\bibliography{./bib/main}
\bibliographystyle{theapa}

\end{document}